\documentclass[runningheads]{llncs}
\usepackage{amsmath}
\usepackage{mathtools}
\DeclarePairedDelimiter{\abs}{\lvert}{\rvert}
\usepackage{amssymb}
\usepackage{xcolor}
\usepackage{tikz}
\tikzset{v/.style={fill=white, draw=black, circle, inner sep=1pt, minimum size=9mm}}
\tikzset{e/.style={draw=black!40,->, >=stealth, line width=0.6mm}}
\tikzset{a/.style={draw,  ->, >=stealth, line width=0.6mm, color=green!60!lightgray}}
\tikzset{vwhite/.style={v, fill=white, draw, line width=0.3mm}}
\tikzset{n/.style={fill=green!50, draw=black, rectangle, inner sep=1pt, minimum size=9mm}}
\tikzset{font={\fontsize{11pt}{12}\selectfont}}
  \usetikzlibrary{decorations.pathreplacing}
\usepackage{graphicx}
\usepackage{float}
\usepackage[caption=false]{subfig}
\usepackage[misc]{ifsym} 
\newcommand{\start}{\textit{start}}
\newcommand{\target}{\textit{target}}
\newcommand{\ord}{\textit{ord}}
\newcommand{\length}{\textit{length}}
\newcommand{\cost}{\textit{cost}}

\newtheorem{prop}{Property}{\itshape}{\rmfamily}

\begin{document}
\title{A Flow Formulation for Horizontal Coordinate Assignment with Prescribed Width}
\titlerunning{Horizontal Coordinate Assignment with Prescribed Width}
%
\author{Michael J\"unger\inst{1} \and Petra Mutzel\inst{2}\orcidID{0000-0001-7621-971X} \and Christiane Spisla\inst{2}\textsuperscript{(\Letter )}}

\institute{University of Cologne, Cologne, Germany\\
\email{mjuenger@informatik.uni-koeln.de}
\and TU Dortmund University, Dortmund, Germany\\
\email{\{petra.mutzel,christiane.spisla\}@cs.tu-dortmund.de}
}
\authorrunning{J\"unger et al.}

%

\maketitle              
\begin{abstract}
We consider the coordinate assignment phase of the well known Sugiyama framework for drawing directed graphs in a hierarchical style. 
The extensive literature in this area has given comparatively little attention to a prescribed width of the drawing.
We present a minimum cost flow formulation that supports prescribed width and optionally other criteria like lower and upper bounds on the distance of neighboring nodes in a layer or enforced vertical edge segments. 
In our experiments we demonstrate that our approach can compete with state-of-the-art algorithms.

\end{abstract}
\section{Introduction}

The Sugiyama framework~\cite{Sugiyama} is a popular approach for drawing directed graphs. It layouts the graph in a hierarchical manner and works in five phases: Cycle removal, layer assignment, crossing minimization, coordinate assignment and edge routing.
If the graph is not already acyclic, some edges are reversed to prepare the graph for the next phase.
Then each node is assigned to a layer so that all edges point from top to bottom.
After that the orderings of the nodes within each layer are determined.
In the coordinate assignment phase that we consider here, the exact positions of the nodes are fixed.
Finally the edges are layouted, e.g., as straight lines. A good overview over the different phases of the framework can be found in \cite{HNHandbook}.

After the nodes are assigned to layers and the orderings of the nodes within their layers are fixed, the task of the coordinate assignment phase is to compute $x$-coordinates for all nodes. There are several, sometimes contradicting, objectives  in this phase, e.g., short edges, minimum distance between neighboring nodes, straight edges, balanced positions of the nodes between their neighbors in adjacent layers, and few bend points of edges that cross multiple layers.
The criterion ``short edges'' can be handled by
exact algorithms as well as fast heuristics that give pleasant results, possibly also considering other aesthetic criteria.

When it comes to the width of the drawing one usually tries to restrict the maximum number of nodes in one layer, see e.g. \cite{CG}. 
Long edges, i.e.\ edges that span more than two layers, are often split into paths with one dummy node on each intermediate layer.
Healy and Nikolov~\cite{HN} present a branch-and-cut approach to compute a layering that takes 
the influence of the number of dummy nodes on the width into account. 
Jabrayilov et al.~\cite{Jabrayilov} do the same in a mixed integer program that treats the first two phases of the Sugiyama framework simultaneously. 
But still, the maximum number of nodes in one layer does not necessarily define the actual width of the final drawing, as illustrated in Fig.~\ref{motivation}. The main objective of most methods for the coordinate assignment phase is ``short edges'', which often leads to small drawings, but the width of the final layout is not directly addressed.

\begin{figure}[tb]
\begin{minipage}{0.45\textwidth}
\centering
\resizebox{0.35\width}{!}{	
\begin{tikzpicture}
\tikzset{font={\fontsize{20pt}{12}\selectfont}}
\tikzset{v/.style={fill=white, draw=black, circle, inner sep=1pt, minimum size=5mm}}

	\node[v] at (0,4) (a) {};
	\node[v] at (0,2) (b) {};
	\node[v] at (1,2) (c) {};
	\node[v] at (0,0) (d) {};
	\node[v] at (1,0) (e) {};

	\node[v] at (0,-2) (f) {};
	\node[v] at (1,-2) (g) {};
	\node[v] at (1,-4) (h){};

	\draw[e, black] (a)--(b);
	\draw[e, black] (c)--(d);
	\draw[e, black] (g)--(h);
	
	\draw[e, black] (e) -- (f); 
	\node[fill=white, rectangle, minimum size=7mm] at (0.5,-1) (bullets){$\cdots$};
	
	\draw [decorate,decoration={brace,amplitude=10pt, mirror}]
(-0.5,4.2) -- (-0.5,-4.2) node [black,midway, xshift=-8mm] {$k$};

\end{tikzpicture}
} 
\end{minipage}
\hfill
\begin{minipage}{0.45\textwidth}
\centering
\resizebox{0.35\width}{!}{	
\begin{tikzpicture}
\tikzset{font={\fontsize{20pt}{12}\selectfont}}
\tikzset{v/.style={fill=white, draw=black, circle, inner sep=1pt, minimum size=5mm}}

	\node[v] at (0,4) (a) {};
	\node[v] at (0,2) (b) {};
	\node[v] at (1,2) (c) {};
	\node[v] at (1,0) (d) {};
	\node[v] at (2,0) (e) {};

	\node[v] at (2,-2) (f) {};
	\node[v] at (3,-2) (g) {};
	\node[v] at (3,-4) (h){};

	\draw[e, black] (a)--(b);
	\draw[e, black] (c)--(d);
	\draw[e, black] (g)--(h);
	
	\draw[e, black] (e) -- (f); 
	\node[fill=white, rectangle, minimum size=7mm] at (2,-1) (bullets){$\cdots$};
	
	\draw [decorate,decoration={brace,amplitude=10pt, mirror}]
(-0.5,4.2) -- (-0.5,-4.2) node [black,midway, xshift=-8mm] {$k$};

\end{tikzpicture}
} 
\end{minipage}
\caption{In the left picture the horizontal edge length is $k-3$ and the width is $1$, in the right picture the horizontal edge length is $0$ and the width is $k-2$, where $k$ is the number of layers.}
\label{motivation}
\end{figure}
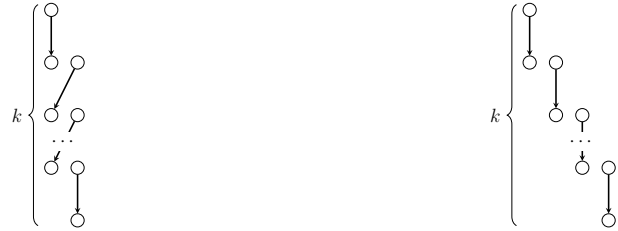

There may be further requirements for the final drawing, such as an aspect ratio in order to make optimal use of the drawing area, or a maximum distance between two nodes on the same layer if they are semantically related. A common request is that inner segments of long edges are drawn as vertical straight lines in order to improve readability.

\noindent
\emph{Related work.} Sugiyama et al.~\cite{Sugiyama} present a quadratic programming formulation that has a combination of two asthetic criteria as objective function, short edges (closeness to adjacent nodes) and a balanced layout (positioning nodes close to the barycenter of their upper and lower neighbors). 
Gansner et al.~\cite{Gansner} give a simpler formulation in which they replace quadratic terms of the form $(x_v - x_u)^2$ by $|x_v-x_u|$ and leave out the balance terms.
The coordinate assignment problem can be interpreted as an instance of the layer assignment problem, and they suggest to apply the network simplex algorithm to an auxiliary graph to obtain a drawing with minimum horizontal edge length.
Given an initial layout, some heuristics sweep through the layers and try to shift the nodes to better positions depending on the fixed $x$-coordinates of their neighbors in adjacent layers, see e.g. \cite{Eades,Sander,Sugiyama}
.
Two fast heuristics that compute coordinates from scratch are presented by
Buchheim et al.~\cite{BJL} and by Brandes and K\"opf~\cite{BK}. Both algorithms draw inner segments of long edges straight 
and aim for a balanced layout with short edges.

\noindent
\emph{Our Contribution.} We formulate the coordinate assignment problem as a minimum cost flow problem that can be solved efficiently. Within this formulation we can fix the maximum width of the final drawing as well as a maximum 
and minimum horizontal distance between nodes in the same layer and we can enforce straightness to some edges. We compute $x$-coordinates such that the total horizontal edge length is minimized subject to these further constraints.

\section{Notation and Preliminaries}

Let $G=(V,E)$ be a directed graph with $|V|=n$ nodes and $|E|=m$ edges. For a directed edge $e=(u,v)$ we denote the start node of $e$ with $\start(e)=u$ and the target node of $e$ with $\target(e)=v$. A \emph{path} $P$ from $u$ to $v$ of length $k$ is a set of edges $\{e_i=(v_i,v_{i+1})\mid i=1,\ldots,k\mbox{ where } u=v_1\mbox{ and }v=v_{k+1}\}$. We also write $u \overset{*}\rightarrow v$. If $v_{k+1}=v_1$ it is called a \emph{cycle}.   
A graph is called a \emph{directed acyclic graph (DAG)} if it has no cycles.
A \emph{layering} $\mathcal{L}$ of a graph assigns every $v\in V$ a \emph{layer} $L_i$, such that $i<j$ holds for every edge $e=(u,v)$ with $\mathcal{L}(u)=L_i$ and $\mathcal{L}(v)=L_j$. The layering is called \emph{proper} if $\mathcal{L}(v)=\mathcal{L}(u)+1$ for every edge~$(u,v)$, i.e., the layers of every pair of adjacent nodes are consecutive.
An edge that violates the latter property is called a \emph{long edge}. Every graph with a layering can be transformed into a graph with a proper layering by subdividing every long edge into a chain of edges.
We denote with $|\mathcal{L}|$ the number of layers and with $|L_i|$ the number of nodes in layer $L_i$.

An \emph{ordering} $\ord$ defines a partial ordering on the nodes of $G$. For every layer~$L_i$ it assigns each node in $L_i$ a number $1\leq j \leq |L_i|$ and we write $u<v$ if $\ord(u)<\ord(v)$. 
We denote with $v_j^i$ the $j$-th node in layer $L_i$.

Given a graph $G$ with a layering $\mathcal{L}$ and an ordering $\ord$ the \emph{horizontal coordinate assignment problem (HCAP)} asks for $x$-coordinates for every node, so that $x(u)<x(v)$ if $u<v$.  
We will restrict ourselves to integer coordinates. The \emph{horizontal length} of an edge $e=(u,v)$ is defined as $\length(e)=|x(v)-x(u)|$ and the \emph{total horizontal edge length} is $\length(E)= \sum_{e \in E} \length(e)$. 
The \emph{width} of the assignment is $\max_{v\in V} x(v) -  \min_{v\in V} x(v)$.
Unless otherwise stated, we mean the horizontal length whenever we talk about the length of an edge.

\emph{HCAP\textsubscript{minEL}} is the variant of HCAP in which we also want to mimimize the total horizontal edge length.


%

We assume familiarity with minimum cost flows. Ahuja et al.~\cite{Ahuja} give a good overview. Let $N=(V_N,E_N)$ be a directed graph with a super source $s$ and a super sink $t$, so for all other nodes the amount of incoming flow equals the amount of outgoing flow. We have lower and upper bounds on the edges and a cost function $\cost:E_N \rightarrow \mathbb{R}$. Let $f$ be a feasible flow.
For a subset of nodes $V^\prime~\subseteq~V_N\setminus~\{s,t\}$ we denote with $f(V^\prime)=\sum_{v\in V^\prime} \sum_{e=(v,w)}f(e)=~\sum_{v\in V^\prime} \sum_{e=(u,v)}f(e)$ the flow through $V^\prime$. For $s$ we define $f(s)$ to be the total amount of flow leaving $s$.
For a subset of edges $E^\prime \subset E_N$ we denote with $f(E^\prime)=\sum_{e\in E^\prime} f(e)$ the flow over $E^\prime$ and with $\cost(E^\prime)=\sum_{e\in E^\prime} \cost(e)$ the cost of $E^\prime$ and with $\cost_f~=~\sum_{e\in E_N}f(e)\cdot \cost(e)$ the total cost of~$f$.  


\section{Network Flow Formulation}
\label{network}

In this section we describe the construction of a network for the horizontal coordinate assignment problem. Given a minimum cost flow in this network we show how to obtain $x$-coordinates for all nodes such that the total horizontal edge length is minimized. 
By a simple modification we can compute $x$-coordinates that give us minimum total horizontal edge length with respect to a given maximum width of the drawing.
The basic idea is that flow represents horizontal distance and we send flow from top to bottom through the layers.

\subsection{Network Construction}
Let $G=(V,E)$ be a DAG with a proper layering $\mathcal{L}$ and an ordering and let $N=(V_N,E_N)$ be the minimum cost flow network. 
For now let us assume that neighboring nodes on a layer should have an equal minimum distance of one and that we have no further requirements concerning the edges.

For every layer $L_i$ with $i \in \{1,\ldots,|\mathcal{L}|\}$ we add nodes $w_0^i,w_1^i,\ldots,w_{|L_i|}^i$ and $z_0^i,z_1^i,\ldots,z_{|L_i|}^i$ to $N$. 
Imagine the node $w_j^i$ placed above the layer $L_i$ and between $v^i_j$ and $v^i_{j+1}$ ($w^i_0$ is placed at the left end and $w^i_{|L_i|}$ at the right end of the layer). The nodes $z^i_j$ are placed in the same way below layer $L_i$. 
Although we do not have a drawing of $G$ at this moment we can still use terms like ``above'' and ``below'' because the layering gives us a vertical ordering of the nodes of $G$ and we can talk about ``left'' and ``right'' because of the given ordering of the nodes in each layer. 
Since we are placing the nodes $w^i_j$ and $z^i_j$ ``between'' the nodes $v^i_j$ and $v^i_{j+1}$ we want to extend the ``$<$'' relation to give a partial ordering on $V \cup V_N$ in the following way: 
$w^i_0 < v^i_1 < w^i_1 < v^i_2 < \cdots < v^i_{|L_i|} < w^i_{|L_i|}$ and  $z^i_0 < v^i_1 < z^i_1 < v^i_2 < \cdots < v^i_{|L_i|} < z^i_{|L_i|}$.
We connect $w^i_j$ to $z^i_j$ with an edge $a^i_j$ that has a lower bound of one and an upper bound of $\infty$ and a cost of zero. 
The flow over these edges will define the distance between $v^i_j$ and $v^i_{j+1}$. 
We denote the set of these edges with $A$. Figure~\ref{arcsets}(a) shows an example.

\begin{figure*}[tb]
\subfloat[]{
\centering
\resizebox{0.7\width}{!}{	
\begin{tikzpicture}[rotate=-90, yscale=-1]

	\node[v, minimum size=10mm] at (0,-1.5) (v2) {$v_{j}$};
	\draw[e] (-2,-1.5) -- (v2);
	\draw[e] (v2) -- (2, -1);
	\draw[e] (v2) -- (2, -2);

	\node[n] at (-1.3, 0.5) (w) {$w_{j-1}^i$};
	\node[n] at (1.3, 0.5) (z) {$z_{j-1}^i$};
	\draw[a] (w) -- node[pos=0.6, left]{\textcolor{black}{$a^i_{j-1}$}}(z) ;
	
	\node[n] at (-1.3, -3.5) (w2) {$w_{j}^i$};
	\node[n] at (1.3, -3.5) (z2) {$z_{j}^i$};
	\draw[a] (w2) -- node[pos=0.4, right]{\textcolor{black}{$a^i_{j}$}}(z2) ;
	
	\draw[a] (w) .. controls +(-0.5,-2) .. (w2);
	\draw[a] (w2) .. controls +(0.5,2) ..(w);
	\node[] at (-2.7, -1.5) (cost1) {\textcolor{black}{$cost(\overleftarrow{bw}^i_j)$}\textcolor{black}{$=cost(\overrightarrow{bw}_j^i)$}\textcolor{black}{$=1$}};
	
	\draw[a] (z) .. controls +(-0.5,-2) ..(z2);
	\draw[a] (z2) .. controls +(0.5,2) .. (z);
	\node[] at (2.7, -1.5) (cost2) {\textcolor{black}{$cost(\overleftarrow{bz}^i_j$)}\textcolor{black}{$=cost(\overrightarrow{bz}_j^i)$}\textcolor{black}{$=2$}};
	
	\node at (0.6, -0.5) (text1) {$\overrightarrow{bz}^i_j$};
	\node at (2, -2.7) (text1) {$\overleftarrow{bz}^i_j$};
	\node at (-2, -0.5) (text1) {$\overrightarrow{bw}^i_j$};
	\node at (-0.6, -2.7) (text1) {$\overleftarrow{bw}^i_j$};
\end{tikzpicture}
}
}
\hfill
\subfloat[]{
\centering\resizebox{0.7\width}{!}{	
\begin{tikzpicture}[rotate=-90, yscale=-1
]
	\node[v, minimum size=10mm] at (0,1) (vk) {$v_j^i$};
	\node[v] at (0,-1) (v2) {$v_{j+1}^i$};
	\node at (0,2.5) (bullets){$\bullet \bullet \bullet$};
	\node[v, minimum size=10mm] at (0,4) (a){};
	\node at (0,-2.5) (bullets2){$\bullet \bullet \bullet$};
	\node[v, minimum size=10mm] at (0,-4) (b){};

	\node[v, minimum size=10mm] at (4,1) (vl) {$v_k^{i+1}$};
	\node[v] at (4,-1) (v3) {$v_{k+1}^{i+1}$};
	\node at (4,-2.5) (bullets3){$\bullet \bullet \bullet$};
	\node[v, minimum size=10mm] at (4,-4) (c){};
	\node at (4,2.5) (bullets4){$\bullet \bullet \bullet$};
	\node[v, minimum size=10mm] at (4,4) (d){};

	\draw[e] (vk) -- (d) node[pos=0.2, left]{$e_1$} ;
	\draw[e] (v2) -- (c)node[pos=0.2, right]{$e_2$};
	
	\draw[e] (a) -- (vl) node[pos=0.8, left]{$e_3$} ;
	\draw[e] (b)-- (v3)node[pos=0.8,right]{$e_4$} ;
	

	\node[n] at (1, 0) (z) {$z^i_j$};
	\node[n] at (3, 0) (w) {$w^{i+1}_k$};
	\draw[a] (z) -- node[left]{\textcolor{black}{$c^i_{jk}$}}(w);
	
\end{tikzpicture}
}
}
\caption{Illustration of edges of the sets (a) $A$, $B$ and (b) $C$. Nodes of $G$ are white circles, nodes of $N$ are green rectangles. Edges of $G$ are gray, edges of $N$ are green.
}
\label{arcsets}
\end{figure*}
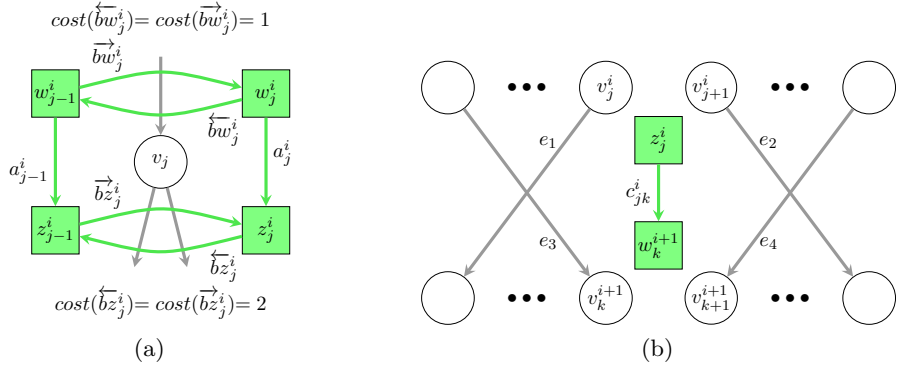

For every layer $L_i$ with $i \in \{1,\ldots,|\mathcal{L}|\}$ and every $j \in \{ 0,\ldots,|L_i| -1\}$ we add edges $\overrightarrow{bw}_j^i=(w_j^i,w_{j+1}^i)$, $\overleftarrow{bw}_j^i=(w_{j+1}^i,w_{j}^i)$, $\overrightarrow{bz}_j^i=(z_j^i,z_{j+1}^i)$ and $\overleftarrow{bz}_j^i=(z_{j+1}^i,z_{j}^i)$ to $N$. 
The lower bound of these edges is zero and the upper bound is $\infty$. The cost of these network edges equals the number of graph edges they ``cross over''. That means, the cost of $\overleftarrow{bw}_j^i$ and $\overrightarrow{bw}_j^i$ equals the number of incoming graph edges of node $v^i_j$ and the cost of $\overleftarrow{bz}_j^i$ and $\overrightarrow{bz}_j^i$ equals the number of outgoing graph edges of~$v^i_j$, see Fig.~\ref{arcsets}(a). 
Positive flow over one of these edges will cause the crossed-over graph edges to have positive horizontal length. We call the set of these edges~$B$.


Now we connect the nodes of neighboring layers. We could add edges between every $z^i_j$ and every $w^{i+1}_k$, but we want to keep the number of edges between layers as small as possible. We add edges only in special situations and will show later that this suffices for correctness.
For every layer~$L_i$ with $i \in \{1,\ldots,|\mathcal{L}|-1\}$ we add edges $c^i_{00}=(z^i_0,w^{i+1}_0)$ and $c^i_{|L_i||L_{i+1}|}=(z^i_{|L_i|}, w^{i+1}_{|L_{i+1}|})$ to the network with a lower bound of zero, an upper bound of $\infty$ and a cost of zero. 
Additionally we add edges $c^i_{jk}=(z^i_j, w^{i+1}_k)$ if there exist $ e_1,e_2,e_3,e_4 \in E$ with $\start(e_1)=v^i_j$, $\start(e_2)=v^i_{j^\prime}$, where $v^i_{j^\prime}$ is the next node to the right of $v^i_j$ with an outgoing edge and $\target(e_3)=v^{i+1}_k$, $\target(e_4)=v^{i+1}_{k^\prime}$, where $v^{i+1}_{k^\prime}$ is the next node to the right of $v^{i+1}_k$ with an incoming edge
and the following conditions holds:
$\start(e_3)\leq \start(e_1) < \start(e_2)\leq \start(e_4)$ and $\target(e_1)\leq  \target(e_3) <  \target(e_4) \leq  \target(e_2)$.
We call this situation a \emph{hug} between $z^i_j$ and $w^{i+1}_k$.
These edges get a lower bound of zero, an upper bound of~$\infty$, and the cost equals the number of graph edges they cross over:
$\cost(c^i_{jk})=|\{ e=(v^i_p, v^i_q)\in E~|~p\leq j \wedge q\geq k^\prime \mbox{ or } p\geq j^\prime \wedge q\leq k\} |$. Like the edges of $B$, flow on edges of this kind will cause horizontal length and we denote the set of all~$c^i_{jk}$ by $C$. Figure~\ref{arcsets}(b) illustrates a hug situation.


Finally we add a super source $s$ and a super sink $t$ to the network. We connect $s$ with every $w^1_j$, $j\in \{ 1,\ldots, |L_1|\}$ and $t$ with every $z^{|\mathcal{L}|}_k$, $k\in \{ 1,\ldots,|L_{|\mathcal{L}|}|\}$. 
These edges get a lower bound of zero, an upper bound of $\infty$ and a cost of zero. Figure~\ref{example} shows a complete example network.
If it is clear from the context which layer or which node is meant, we omit the node subscripts and superscripts. 

\begin{figure}[tb]
\centering
\includegraphics[scale=0.2]{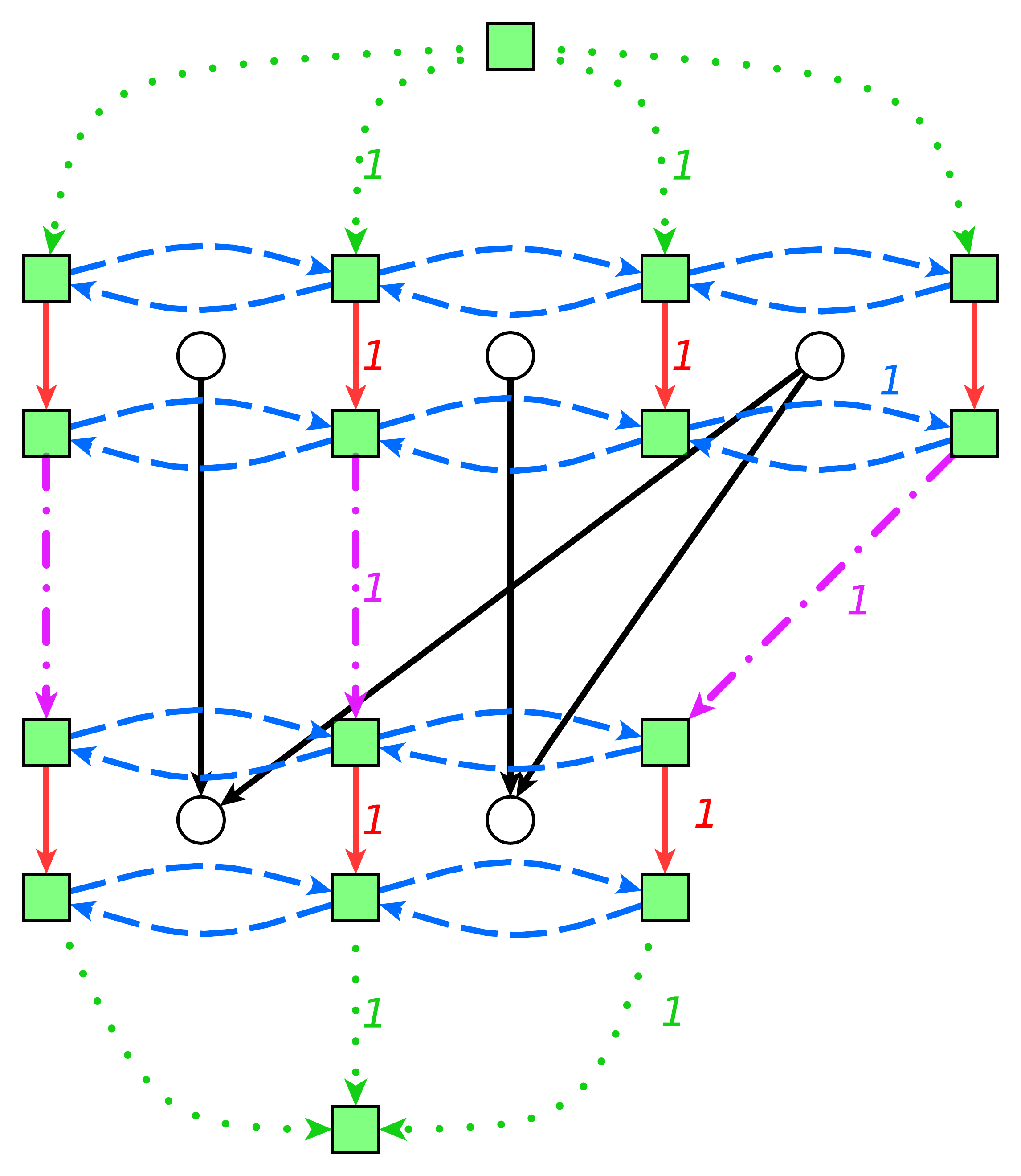}
\caption{An example network with underlying graph. Edges of the set $A$ are red (solid), edges of the set $B$ are blue (dashed) and edges of the set $C$ are purple (dashdotted). The numbers along the edges denote the flow, unlabeled edges carry no flow.}
\label{example}
\end{figure}

\subsection{Obtaining Coordinates and Correctness}

Let $f$ be a feasible flow in the network described above. We observe that $f(a^i_j)=f(w^i_j)=f(z^i_j)$ since $a^i_j$ is the only outgoing edge of $w^i_j$ and the only incoming edge of $z^i_j$. 
We define the $x$-coordinate of a node $v^i_j$ as
\begin{equation} \label{defx}
 x(v^i_j) := \sum_{l=0}^{j-1} f(a^i_l) = \sum_{l=0}^{j-1} f(w^i_l) = \sum_{l=0}^{j-1} f(z^i_l).
\end{equation}
Together with $y(v^i_j)=i$ we get an induced drawing with a feasible coordinate assignment, because for every $v_j, v_k$ within the same layer $x(v_j)<x(v_k)$ if and only if $v_j<v_k$ (since the amount of flow over edges a~$\in A$ is always positive).

Now we want to explain the correspondence between the cost of a flow $f$ and the total horizontal edge length of the resulting drawing. The intuition is, that if flow is sent from the right of $\start(e)$ to the left of $\target(e)$ for some edge $e$, then $\target(e)$ is ``pushed'' to the right because of the additional flow on the left. This results in a horizontal expansion of $e$.
We define for an edge $e=(u,v)\in E$
\begin{eqnarray*}
\overrightarrow{E}(e)&:=&\{ bw \in B\mid \start(bw) < v~\wedge~\target(bw) > v \}\\
&&\cup~\{ bz \in B\mid \start(bz) < u~\wedge~\target(bz) > u \}\\
&&\cup~\{ c \in C\mid\start(c) < u~\wedge~\target(c) > v \}
\end{eqnarray*}
as the set of network edges that start to the left of $e$ and end to the right of $e$, thus cross over $e$ from left to right. 
Analogously the set of network edges that cross over a graph edge from right to left is 
\begin{eqnarray*}
\overleftarrow{E}(e)&:=&\{ bw \in B\mid\start(bw) > v~\wedge~\target(bw) < v \}\\
&&\cup~\{ bz \in B\mid\start(bz) > u~\wedge~\target(bz) < u \}\\
&&\cup~\{ c \in C\mid\start(c) > u~\wedge~\target(c) < v\}.
\end{eqnarray*}

We make the following observations: 
\begin{prop} \label{cost=E<->}
$\cost(g)=| \{ e\in E\mid g\in \overrightarrow{E}(e) \}| + |\{ e\in E\mid g\in \overleftarrow{E}(e) \} |\; \forall g \in B\cup C$.
\end{prop}

\begin{prop} \label{floww=flowz}
Because of the flow conservation rule we have \\$\sum_{j=0}^{|L_i|} f(w^i_j) = \sum_{j=0}^{|L_i|} f(z^i_j)$ for all $i\in \{1,\ldots,|\mathcal{L}|\}$ and \\$\sum_{j=0}^{|L_i|} f(w^i_j) = \sum_{j=0}^{|L_{k}|} f(w^k_j)=f(s)$ for all $i,k\in \{1,\ldots,|\mathcal{L}|\}$.
\end{prop}

\begin{prop} \label{s<width}
The width of the induced 
drawing is \\
$\max_{1\leq i \leq |\mathcal{L}|} \left( \sum_{j=1}^{|L_i|-1} f(w_j^i) \right)\leq f(s)$. 
\end{prop}

\begin{prop} \label{w=z+<->}
Let $e=(v^i_j, v^{i+1}_k)$ be an edge. Then \\$\sum_{w^{i+1}_l < v^{i+1}_k} f(w^{i+1}_l) = \sum_{z^{i}_l < v^{i}_j} f(z^{i}_l) + f(\overleftarrow{E}(e)) - f(\overrightarrow{E}(e))$.
\end{prop}

The last property is illustrated in Fig.~\ref{prop4}. The total flow that reaches all $w^{i+1}_j$ that are to left of $\target(e)$ comes from the $z^i_j$ that are to the left of $\start(e)$ and from nodes that are to the right of $\target(e)$ or $\start(e)$. Flow from the latter nodes has to pass over $e$ from right to left.
Flow from a node $z^i_j$ that is to the left of $\start(e)$ and does not enter one of the $w^{i+1}_j$ left of $\target(e)$ has to pass over $e$ from left to right.
 
\begin{lemma}
For a feasible flow $f$ and the induced drawing $\cost_f\ge\length(E)$ holds.
\end{lemma}

\begin{proof}
	Let $e=(v^i_j,v^{i+1}_k)$ be an edge of $G$.	
	The length of $e$ is $\length(e)=|x(v^i_j) - x(v^{i+1}_k)|$ and together with (\ref{defx}) we have 
\begin{eqnarray*}
\length(e)&=&\abs[\Big]{ \sum_{l=0}^{j-1} f(z^i_l) - \sum_{l=0}^{k-1} f(w^{i+1}_l)} \\			
&=&\abs[\Big]{\sum_{z^{i}_l < \start(e)} f(z^i_l) - \sum_{w^{i+1}_l < \target(e)} f(w^{i+1}_l)} \\
&=&\abs[\Big]{f(\overrightarrow{E}(e)) - f(\overleftarrow{E}(e))}\qquad\mbox{(by Property~\ref{w=z+<->}).}
\end{eqnarray*}
Therefore we have for the total edge length
\begin{eqnarray*}
\length(E)&=&\sum_{e\in E} \abs[\Big]{f(\overrightarrow{E}(e)) - f(\overleftarrow{E}(e))}\\
&\le&\sum_{e\in E} \left(\abs[\Big]{f(\overrightarrow{E}(e))} + \abs[\Big]{f(\overleftarrow{E}(e)}\right)\\
&=&\sum_{e\in E} \left( f(\overrightarrow{E}(e)) + f(\overleftarrow{E}(e)\right)\\
&=&\sum_{g \in E_N} f(g) \cdot \abs{\{ e\in E~|~g\in \overrightarrow{E}(e) \} \cup \{ e\in E~|~g\in \overleftarrow{E}(e) \}} \\
&=&\sum_{g\in E_N} f(g)\cdot \cost(g)\qquad\mbox{(by Property~\ref{cost=E<->})}   \\
&=&\cost_f.
\end{eqnarray*}	
\qed
\end{proof}

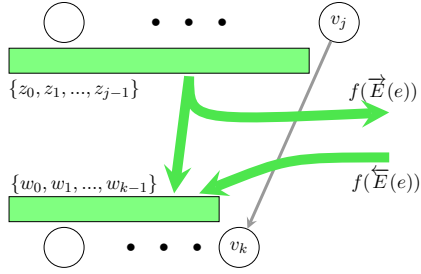
\begin{figure}[tb]
\centering\resizebox{0.665\width}{!}{	
\begin{tikzpicture}[rotate=-90, xscale=0.9]

	\node[v, minimum size=8mm] at (0,2) (l) {$v_j$};
	\node[v, minimum size=8mm] at (5, 0) (r) {$v_k$};
	\draw[e] (l) -- (r);
	
	\node at (0, -1) (bullets){$\bullet \quad \bullet\quad \bullet$};
	\node[v, minimum size=8mm] at (0,-3.5) (a) {};
	\node at (5, -1.5) (bullets2){$\bullet \quad \bullet\quad \bullet$};
	\node[v, minimum size=8mm] at (5,-3.5) (b) {};
	
	\node[rectangle,draw, fill=green!50, minimum height=5mm, minimum width=60mm] at (0.85,-1.6) (w) {};
	\node[rectangle,draw, fill=green!50, minimum height=5mm, minimum width=42mm] at (4.15,-2.5) (z) {};
	\node[black] at (1.5, -3.3) {$\{z_0,z_1,...,z_{j-1}\}$}; 
	\node[black] at (3.6, -3.1) { $\{w_0,w_1,...,w_{k-1}\}$}; 
	
	\draw[a, line width=2mm, green!60!lightgray] (1.2,-1)--(3.8,-1.3);
	\draw[a, line width=2mm, green!60!lightgray] (1.2,-1).. controls +(1,-0).. node[black,above, pos=0.99]{$f(\overrightarrow{E}(e))$} (2,3);
	\draw[a, line width=2mm, green!60!lightgray] (3,3).. controls +(0,-2).. node[black, below, pos=0.01]{$f(\overleftarrow{E}(e))$} (3.8,-0.8);

\end{tikzpicture}
}
\caption{Illustration of Property~\ref{w=z+<->}. The rectangles represent all network nodes to the left of $v_j$ and $v_k$, respectively. The thick arrows represent the flow of several edges.
}
\label{prop4}
\end{figure}
\newpage
\begin{lemma}
Let $\Gamma$ be a drawing of $G$.
There exists a flow $f$ that induces $\Gamma$ and whose cost is equal to the total edge length of $\Gamma$.
\end{lemma}

\begin{proof}
If necessary, we set $x(v):=x(v) - \min_{v\in V}x(v)$ so that the smallest $x$-coordinate is zero. That gives us an equivalent drawing.
We construct the flow~$f$ as follows: 
Let $\omega$ be the width of $\Gamma$.
We send $\omega$ units of flow from $s$ to $t$, so that the $k$-th unit takes the path $P_k=s \overset{*}\rightarrow w^1_{j_1} \overset{*}\rightarrow w^2_{j_2}\overset{*}\rightarrow\cdots\overset{*}\rightarrow w^{|\mathcal{L}|}_{j_{|\mathcal{L}|}} \overset{*}\rightarrow t$, where $w^i_{j_i}$ is chosen so that $x(v^i_{j_i+1})\geq k$ and $x(v^i_{j_i}) < k$ ($w^i_{j_i}=w^i_0$, if $x(v^i_1) \geq k$ and $w^i_{j_i}=w^i_{|L_i|}$, if $x(v^i_{|L_i|} < k$).
That means we send the $k$-th unit through the $k$-th ``column'' of $\Gamma$.
This is always possible, because of the subpaths $w^i_{j_i} \rightarrow z^i_{j_i} \overset{*}\rightarrow z^i_0 \rightarrow w^{i+1}_0 \overset{*}\rightarrow w^{i+1}_{j_{i+1}}$. 
So for every $v$ there are $x(v)$ units of flow that pass by to the left of $v$, thus giving us correct coordinates for all nodes.

We define $E^i_{k}:=\{ e=(v^i_j,v^{i+1}_l)\in E\mid x(v^i_j) < k \text{ and }x(v^{i+1}_l)\geq k \} \cup \{ e\in~E\mid x(v^i_j) \geq k \text{ and }x(v^{i+1}_l)< k \}$, i.e. all edges that cross over the $k$-th column between $L_i$ and $L_{i+1}$.
We show that there exists a path $P_k$ that produces the same cost as the number of graph edges that cross over the $k$-th column in total, that is $\cost(P_k)=\sum_{i=1}^{|\mathcal{L}|-1}\abs{E^i_k}$. Then we have $\sum_{k=1}^\omega \cost(P_k)=\sum_{k=1}^\omega \sum_{i=1}^{|\mathcal{L}|-1}\abs{E^i_k}=\length(E)$ and we have proven the lemma. 

It suffices to focus on the subpath $P^i_k$ from $z=z^i_{j_i}$ to $w=w^{i+1}_{j_{i+1}}$ between two consecutive layers. Notice that network edges $(s, w^1)$, $(w^i,z^{i})$ and $(z^{|\mathcal{L}|},t)$ do not contribute to the cost of the flow. For better readability we denote the nodes of $L_i$ with $u_j$ and the nodes of $L_{i+1}$ with $v_j$ and we omit the superscripts. If not stated otherwise we use $z_j$ for $z^i_j$ and $w_j$ for $w^{i+1}_j$. We construct $P^\prime=P^i_k$ so that $\cost(P^\prime)=|E^\prime|=|E^i_{k}|$.

\medskip
\noindent
\textbf{Case 1:} There exists no edge $e$ with $\start(e)<z$ and $\target(e)<w$. \\
That means every edge $e$ with $\start(e)<z$ has $\target(e)>w$, and if $\target(e)<w$ then $\start(e)>z$.
Then we set $P^\prime=z \rightarrow z_{j_i -1}\overset{*}\rightarrow z_0 \rightarrow w_0 \rightarrow w_1 \overset{*}\rightarrow w$. 
For every $u_j<z$ with $p$ outgoing edges $P^\prime$ uses exactly one $\overleftarrow{bz}$ with cost $p$. All these edges are in $E^\prime$. For every $v_j<w$ with $q$ incoming edges we use exactly one $\overrightarrow{bw}$ with cost $q$. Again these edges are in $E^\prime$. So $\cost(P^\prime)=|E^\prime|$, since there are no other edges in $E^\prime$.

\medskip
\noindent
\textbf{Case 2:} There exists no edge $e$ with $\start(e)>z$ and $\target(e)>w$.\\
Arguing like in Case 1, we set $P^\prime=z \overset{*}\rightarrow z_{|L_i|} \rightarrow w_{|L_{i+1}|} \overset{*}\rightarrow w$. 
As before the cost of $P^\prime$ equals $|E^\prime|$.

\medskip
\noindent
\textbf{Case 3:} There exists an edge $e_l$ with $\start(e_l)<z$ and $\target(e_l)<w$ and another edge $e_r$ with $\start(e_r)>z$ and $\target(e_r)>w$.\\
Let $e_l$ be the edge with the biggest $x(\start(e))$ of all edges $e$ with $\start(e)<z$ and $\target(e)<w$, and let $e_r$ be the edge with the smallest $x(\start(e))$ of all edges $e$ with $\start(e)>z$ and $\target(e)>w$.

\smallskip
\noindent
\textbf{Case 3.1:} There is at least one node $u^\prime$ with outgoing edges and \\$\start(e_l)<u^\prime < z$.\\
Let $u_g=\start(e_l)$ and $v_{g^\prime}=\target(e_l)$. We know $v_{g^\prime}<w$. Let $v_{h^\prime}$ be the first node to the right of $v_{g^\prime}$ with an edge $e_{r^\prime}=(u_h, v_{h^\prime})$ and $u_h>u_g$. Such a node does exist, since we have $e_r$. 
Notice that $v_{h^\prime}$ might be to the right of $w$. 

Then we have a hug: Set $e_1=e_l$, set $e_2$ to one outgoing edge of $u_{g+1}$ (or the next node to the right of $u_g$, which has an outgoing edge), $e_4=e_{r^\prime}$ and set $e_3$ to one incoming edge of $v_{h^\prime-1}$ (or the next one to the left of $v_{h^\prime}$), see Fig.~\ref{case3}.
Notice that $e_1$ may coincide with $e_3$ and $e_2$ with $e_4$.

We have $\start(e_3) \leq  \start(e_1)$, because we chose $e_1=e_l$ with the biggest $x(\start(e))$ and
 $v_{h^\prime}$ is the first node to the right of $v_{g^\prime}$ with an adjacent node to the right of $u_g$. 
So every node between $v_{g^\prime}=\target(e_1)$ and $v_{h^\prime}$, including $v_{h^\prime-1}=\target(e_3)$, can only have adjacent nodes to the left of $u_g=\start(e_1)$.
It is clear that $\start(e_1) < \start(e_2)$ and $\start(e_2) \leq \start(e_4)$, since $\start(e_4)=u_h>u_g$. 
By choice of $e_1$, $e_3$ and $e_4$  $\target(e_1) \leq  \target(e_3) <  \target(e_4)$ holds.
We know that $\target(e_2)>w$ because there is at least one node between $\start(e_1)$ and $z$ whose outgoing edges have to end to the right of $w$ because of the choice of $e_1$.
If $\target(e_4) > \target(e_2)$ then $e_2$ would have been chosen for $e_{r^\prime}$ and therefore for $e_4$. So $\target(e_4) \leq \target(e_2)$ also holds.
So there exists $c_{g(h^\prime -1)}\in E_N$ and we set $P^\prime=z \overset{*}\rightarrow z_g \rightarrow w_{h^\prime -1} \overset{*}\rightarrow w$.

Now for the cost. A subset of $E^\prime$ are the edges $e$ with $z_g < \start(e) <z$ and $\target(e)>w$, which are covered by the $\overleftarrow{bz}$ of $P^\prime$. 
 
Now we have two options. First, if $w_{h-1}>w$ then 
all edges $e$ with $\start(e)<~z_g$ and $\target(e)>w_{h^\prime -1}$ are covered by $c_{g(h^\prime -1)}$ and
the remaining edges $e$ with $\start(e)<~z_g$ and $w<\target(e)<w_{h^\prime -1}$ are covered by the $\overleftarrow{bw}$ of $P^\prime$.
Edges~$e$ with $\start(e)>z>u_g$ and $\target(e) < w <w_{h^\prime}$ are also covered by $c_{g(h^\prime -1)}$.
There cannot be any edge $e$ with $z<\start(e)$ and $w<\target(e)< w_{h^\prime -1}$ or $z_g<\start(e)<z$ and $\target(e)<w$, which would be crossed over by two different edges of $P^\prime$, due to the choice of edges $e_1$ to $e_4$.

Second, if $w_{h-1}<w$ then 
$c_{g(h^\prime -1)}$ covers all edges $e$ with $\start(e)<z_g$ and $\target(e)>w>w_{h^\prime -1}$ and all edges $e$ with $\start(e)>z>z_g$ and $\target(e) < w_{h^\prime -1}<w$.
Edges $e$ with $\start(e)>z$ and $w_{h^\prime -1} < \target(e) <w$ are covered by the $\overrightarrow{bw}$.
Again there are no edges that are crossed over twice by $P^\prime$ due to the choice of $e_1$ to $e_4$.
And there are no edges in $E^\prime$ that are not covered by some edge of $P^\prime$.

\smallskip
\noindent
\textbf{Case 3.2:} $\start(e_l)$ is the next node to the left of $z$ with outgoing edges, but there is at least one node $u^\prime$ with outgoing edges and $\start(e_r)>u^\prime > z$.\\
This case is analogous to Case 3.1.

\smallskip
\noindent
\textbf{Case 3.3:} $\start(e_l)$ is the next node to the left of $z$ with outgoing edges and $\start(e_r)$ is the next node to the right of $z$ with outgoing edges.\\
Let $e_l=(u_g,v_{g^\prime})$ and $v_{h^\prime}$ be the first node right of $v_{g^\prime}$ with an adjacent node $u_h>u_g$. 
Again we have a hug. Set $e_1=e_l$, $e_2=e_r$, $e_3$ to an incoming edge of~$v_{h^\prime -1}$ (or a lower node, if necessary) and $e_4=(u_h,v_{h^\prime})$. 

With the same arguments as in Case 3.1 we convince ourselves that $e_1$, $e_2$, $e_3$ and $e_4$ are indeed a hug and we have $c_{j_i(h^\prime -1)}$.
We set $P^\prime=z \rightarrow w_{h^\prime -1} \overset{*}\rightarrow w$. As before $\cost(P^\prime)=|E^\prime|$.
\qed
\end{proof}

\begin{theorem}
A minimum cost flow in the network described above solves \\
HCAP\textsubscript{minEL}.
\end{theorem}
\begin{proof}
Lemma 1 and Lemma 2.
\qed
\end{proof}

\begin{figure}[tb]
\centering\resizebox{0.5\width}{!}{	
\begin{tikzpicture}[yscale=0.8]
\tikzset{
  font={\fontsize{15pt}{12}\selectfont}}
  
	\node[v] at (0,2) (u1){};
	\node[v] at (3,2) (ug){$u_g$};
	\node[v] at (6,2) (u3){};
	\node[v] at (9,2) (u4){};
	\node[v] at (12,2) (uh){$u_h$};
	\node[v] at (15,2) (u6){};
	
	\node[v] at (0,-6) (v1){};
	\node[v] at (3,-6) (vg){$v_{g^\prime}$};
	\node[v] at (6,-6) (v3){};
	\node[v] at (9,-6) (v4){};
	\node[v] at (12,-6) (vh){$v_{h^\prime}$};
	\node[v] at (15,-6) (v6){};
	
	\draw[e] (u1)--(v3);
	\draw[e, black] (u1)--node[pos=0.8,right=2mm]{$e_3$}(v4);
	\draw[e, black] (ug)--node[pos=0.2,right]{$e_1$}(vg);
	\draw[e, black] (u3)--node[pos=0.25,right=3mm]{$e_2$}(v6);
	\draw[e] (u4)--(v1);
	\draw[e, black] (uh)-- node[pos=0.85, left]{$e_4$}(vh);	
	\draw[e] (u6)--(vg);

	\node[n] at (4.5,1.3) (n2) {};
	\node[n] at (7.5,1.3) (n3) {$z$};
	
	\node[n] at (7.5,-5.3) (n8) {$w$};
	\node[n] at (10.5,-5.3) (n9) {};
	
	\draw[a] (n3)--(n2);
	\draw[a] (n2)--(n9);
	\draw[a] (n9)--(n8);
\end{tikzpicture}
}
\caption{Case 3.1. Only relevant network nodes and edges are depicted. Edges that participate in the hug are black.}
\label{case3}
\end{figure}
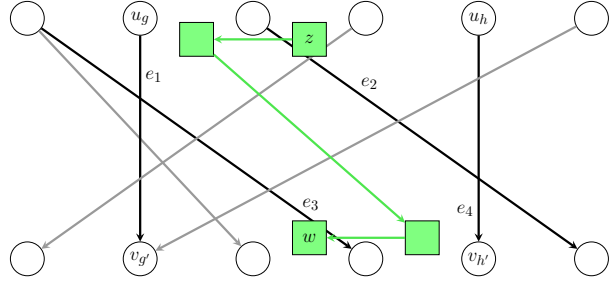

For controlling the maximum width of the drawing we make use of Property~\ref{s<width}, which states that the width of the drawing is at most the flow leaving $s$. We can add an additional node $s^\prime$ and an edge $(s,s^\prime)$ to $N$ and replace all edges of the form $(s, w^1_j)$ with $(s^\prime, w^1_j)$. Now we can limit the maximum width of the drawing by setting the upper bound of $(s,s^\prime)$ to an appropriate value.

Further constraints can be modelled by manipulating the network. By adjusting the lower and upper bounds of edges $a\in A$ we can realize minimum and maximum distances between two neighboring nodes on the same layer.
By removing every $g \in \overleftarrow{E}(e) \cup \overrightarrow{E}(e)$ from the network, we can enforce the edge $e$ to be drawn vertically.

%
%
%
%
 
\section{Experimental Results}

\begin{figure}[tbh]
\centering
\includegraphics[scale=0.9]{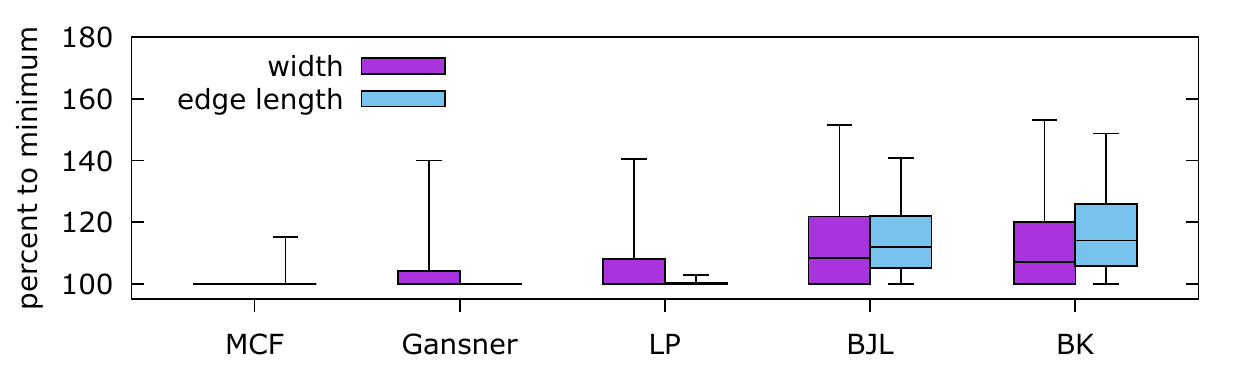}

\caption{Width and total edge length produced by \textsf{MCF}, \textsf{Gansner}, \textsf{LP}, \textsf{BJL} and \textsf{BK} relative to minimum width, resp. edge length.}
\label{widthLength}
\end{figure}

\begin{figure}[tbh]
\centering
\includegraphics[scale=0.9]{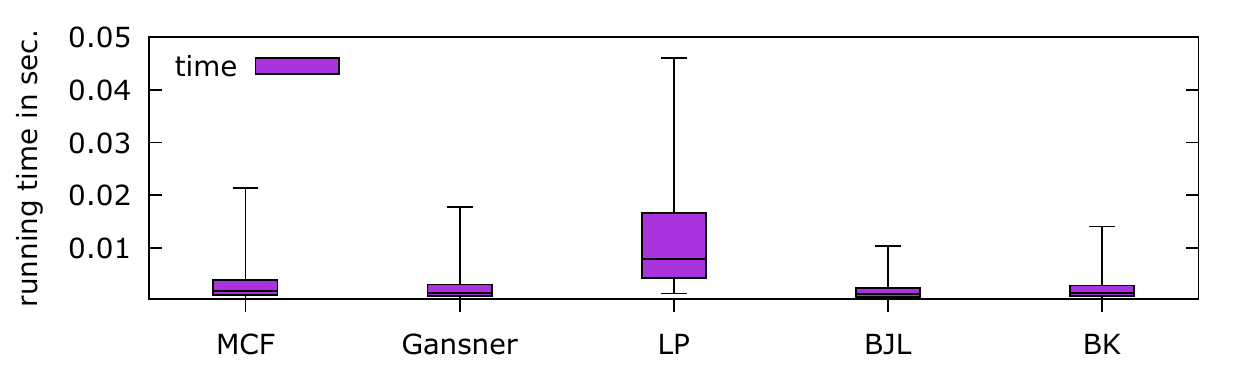}

\caption{Running time for \textsf{MCF}, \textsf{Gansner}, \textsf{LP}, \textsf{BJL} and \textsf{BK}.}
\label{time}
\end{figure}

\begin{figure}[tbh]
\subfloat[]{
\centering
\includegraphics[scale=0.52]{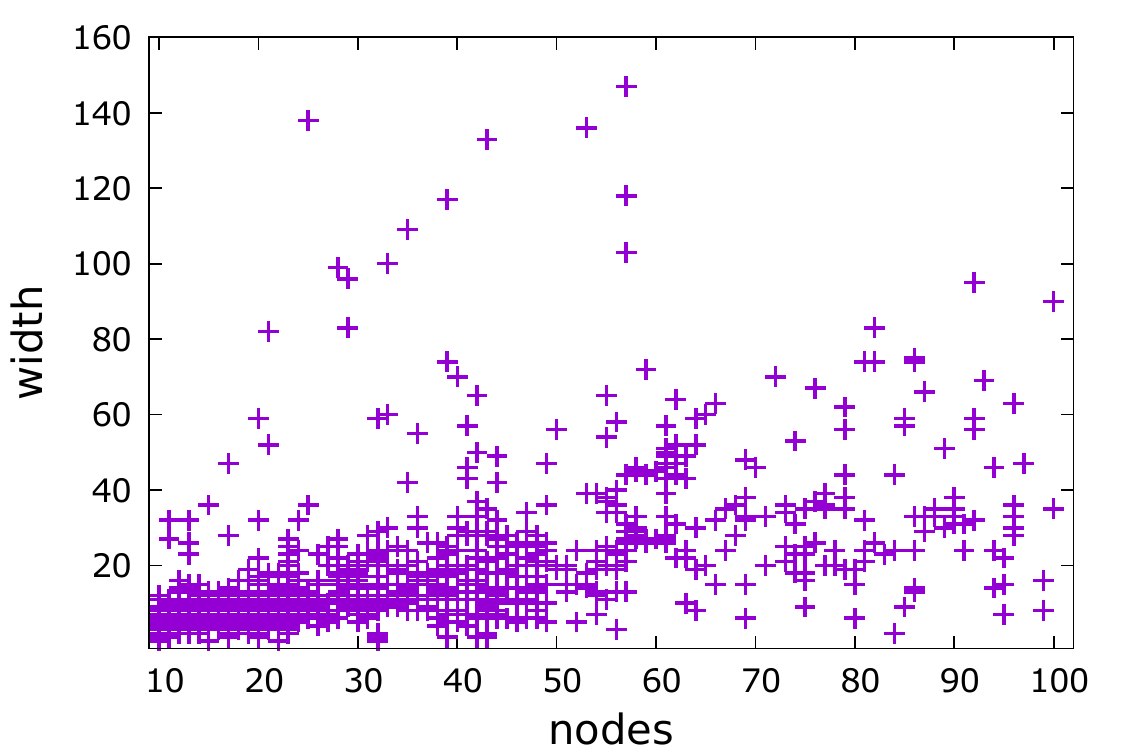}
}
\hfill
\subfloat[]{
\includegraphics[scale=0.52]{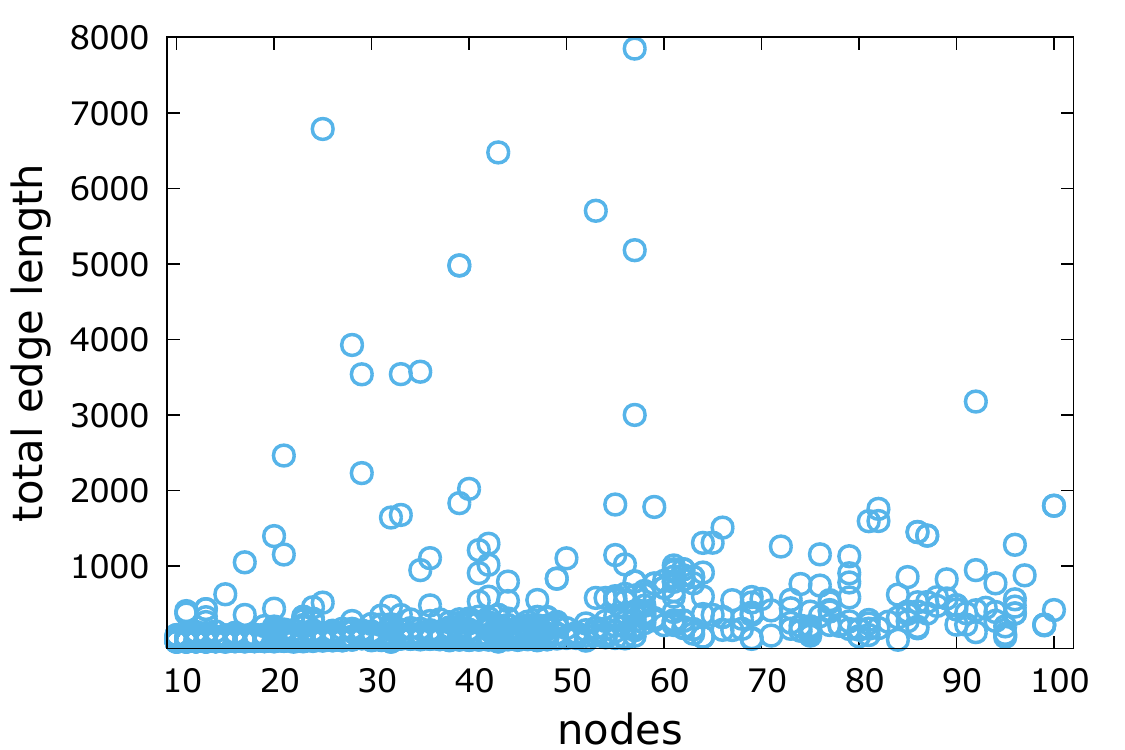}
}
\caption{Absolute values of (a) width and (b) total edge length for \textsf{MCF}.}
\label{abs}
\end{figure}

In our experiment we want to demonstrate that we are able to restrict the width of the drawing without paying too much in terms of total (horizontal) edge length and time.

We implemented the algorithm from Sect.~\ref{network}, which we will call \textsf{MCF} within the Open Graph Drawing Framework~\cite{OGDF} (OGDF) and used the OGDF network simplex software to solve the minimum cost flow problem. We also implemented the approach of Gansner et al.~\cite{Gansner} (\textsf{Gansner}) that also uses the network simplex algorithm. 
Additionally we use three other OGDF methods: an ILP that also takes balancing the nodes between their neighbors into account (\textsf{LP}), the algorithm of Buchheim, J\"unger and Leipert~\cite{BJL} (\textsf{BJL}) and the algorithm of Brandes and K\"opf~\cite{BK} (\textsf{BK}). 
All algorithms
draw inner segments of long edges as vertical lines, since this is generally desirable for good readability. \textsf{MCF} is configured to compute a layout with mimimum edge length with respect to minimum possible width  and \textsf{Gansner} computes coordinates that minimize the total edge length regardless of width.
We used a subset of the AT\&T graphs from \url{www.graphdrawing.org/data.html} consisting of 1277 graphs with 10 to 100 nodes as our test set.

The test was run on an Intel Xeon E5-2640v3 2.6GHz CPU with 128 GB RAM.

Figures~\ref{widthLength}, \ref{time}, and \ref{abs} show the results. The whiskers in Fig.~\ref{widthLength} and \ref{time} cover 95\% of the data and outliers are omitted for better readability. Figure~\ref{abs} shows absolute values for \textsf{MCF} and Figure~\ref{exampleDrawings} displays three example drawings.

In Figure~\ref{widthLength} the resulting total edge length and width of the drawings are depicted relative to the minima that are computed by \textsf{Gansner} and \textsf{MCF}, respectively. We see that \textsf{MCF} still achieves good results in terms of total edge length, even though it has the restriction of meeting the minimum width. The total edge length of drawings computed with \textsf{MCF} is on average 2.2\% over the minimum, while drawings produced with \textsf{Gansner} have on average a width that is 8.9\% over the minimum. In an extreme example with minimum width 1, \textsf{Gansner} results in width 15.

Figure~\ref{time} shows the running time in seconds. 
\textsf{MCF} (4.9 milliseconds on average) is a bit slower than \textsf{Gansner} (3.9 milliseconds on average).
The fastest algorithm on average is \textsf{BJL} with 2.5 milliseconds.

\begin{figure*}[tb]
\subfloat[]{
\centering
\includegraphics[scale=0.35]{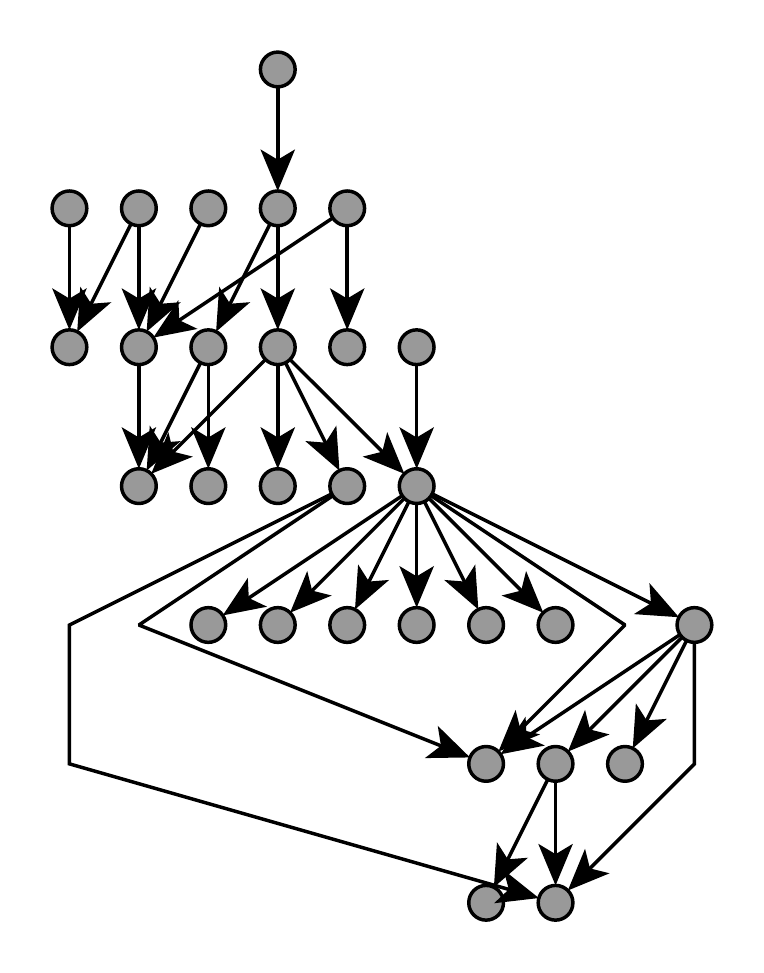}
}
\hfill
\subfloat[]{
\centering
\includegraphics[scale=0.35]{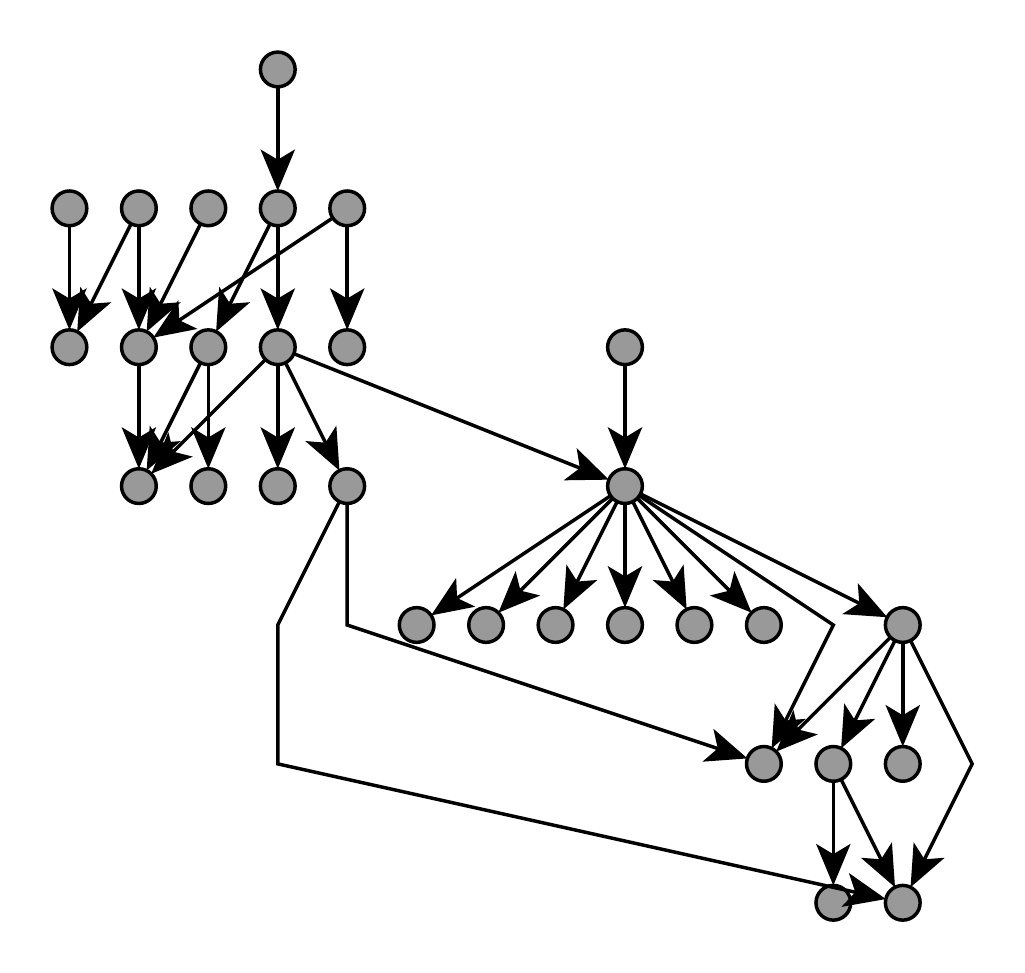}
}
\hfill
\subfloat[]{
\centering
\includegraphics[scale=0.35]{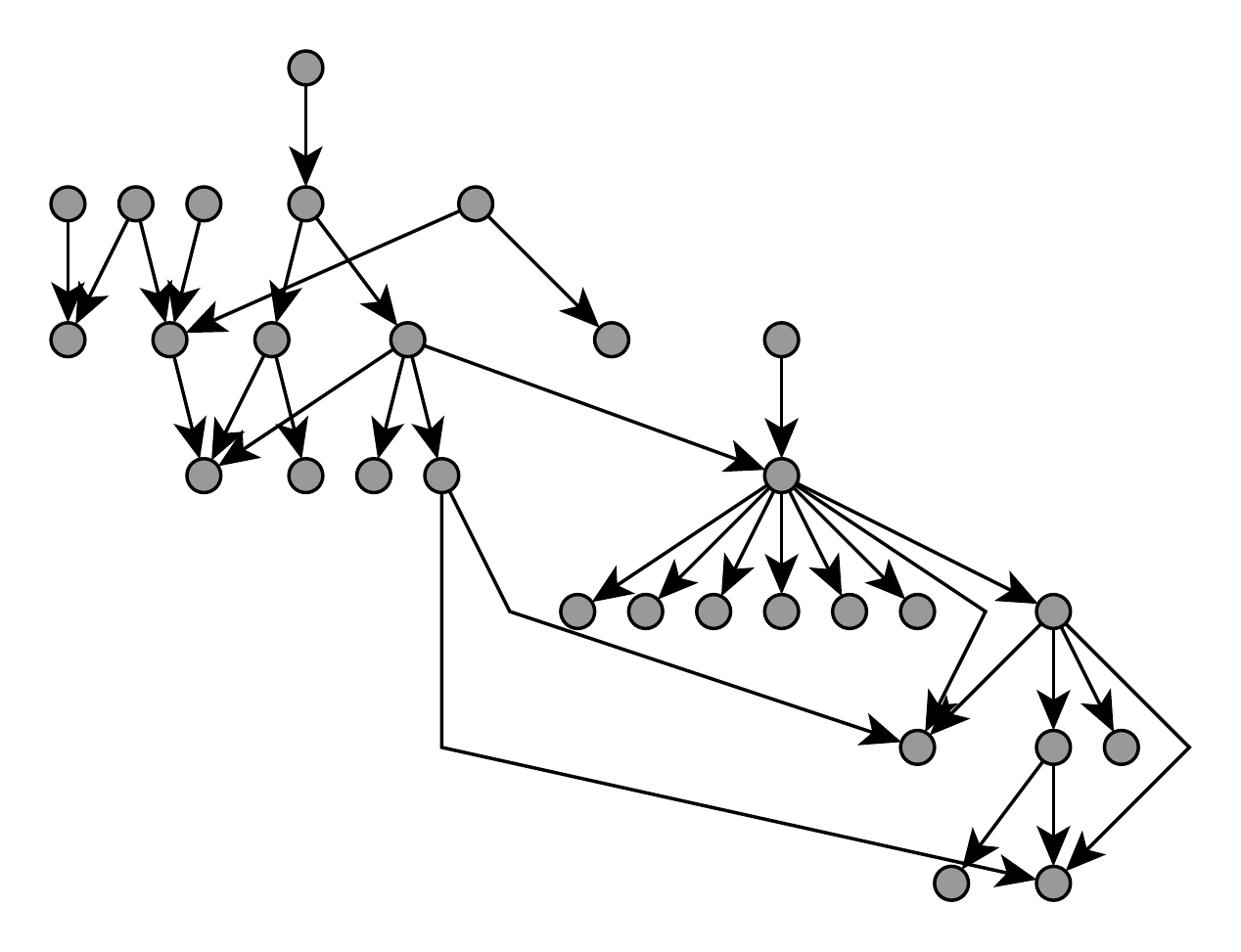}
}
\caption{Example drawings of a graph with 29 nodes and 33 edges. \newline (a) \textsf{MCF}: width: 9, edge length: 58. (b) \textsf{Gansner}: width: 13, edge length: 54. \newline (c) \textsf{BK}: width: 16.5, edge length: 63.5.}
\label{exampleDrawings}
\end{figure*}

\section{Conclusion}
We presented a minimum cost flow formulation for the coordinate assignment problem that minimizes the total edge length with respect to several optional criteria like the maximum width or lower and upper bounds on the distance of neighboring nodes in a layer. 
In our experiments we showed that our approach can compete with state-of-the-art algorithms.

%
%
 \bibliographystyle{splncs04}
 \bibliography{flowCA}

\end{document}